\documentclass[10pt, letterpaper]{article}
\usepackage[T1]{fontenc}
\usepackage{xspace}
\usepackage[colorlinks=true,linkcolor=blue,citecolor=blue]{hyperref}
\usepackage{physics}
\usepackage{authblk}
\usepackage{mathtools}
\usepackage{enumitem}
\usepackage{multirow,array}
\usepackage{amsmath,amsfonts,amssymb}
\usepackage{amsthm}
\usepackage[usenames,dvipsnames,svgnames,table]{xcolor}
\usepackage{braket}
\usepackage{wrapfig}
\usepackage{tocloft}
\usepackage{stmaryrd}
\usepackage[utf8]{inputenc}
\usepackage{graphicx}
\usepackage{setspace}
\usepackage{changepage}
\usepackage{hyphenat}
\usepackage{float}
\usepackage[linesnumbered,ruled,vlined]{algorithm2e}
\usepackage{setspace}
\usepackage{etoolbox}
\usepackage{tcolorbox}
\usepackage{tikz,tikz-qtree}
\usepackage{stackengine}
\usepackage{cancel}
\usepackage{complexity}
\usepackage[labelfont=bf]{caption}
\usepackage[labelfont=bf]{subcaption}
\stackMath
\usepackage{ytableau}
\ytableausetup{boxsize=0.5em, aligntableaux=center}
\usepackage[capitalise,nameinlink]{cleveref}
\usepackage{url}

\newlang{\StateHSP}{StateHSP}
\newlang{\HSP}{HSP}
\newlang{\negl}{negl}

\newcommand{\complexityclass}[1]{{\mathsf{#1}}\xspace}
\renewcommand{\QMA}{\complexityclass{QMA}}
\renewcommand{\QCMA}{\complexityclass{QCMA}}

\def\C{{\mathbb{C}}} % Complex Numbers
\def\Z{{\mathbb{Z}}} % Integers
\newcommand{\bbS}{\mathbb{S}} % symmetric group

\renewcommand{\E}{\mathop{\mathbb{E}\/}}

\newcommand{\Ber}{\mathrm{Ber}} %Bernoulli
\renewcommand{\rank}{\mathrm{rank}}

\NewCommandCopy{\dashl}{\l} % Polish character ł
\renewcommand{\l}{\ell} % nicer-looking than l

\renewcommand{\tilde}{\widetilde}

\newcommand{\B}{\{0,1\}}

\newcommand{\calM}{\mathcal{M}}

\DeclareMathOperator{\MSet}{MSet}
\DeclareMathOperator{\Symop}{Sym}
\DeclareMathOperator{\Haarop}{Haar}

\DeclareMathOperator{\Irrop}{Irr}

\newcommand{\Sym}[2]{\Symop^{#1}\left(#2\right)}
\newcommand{\Haar}[1]{\Haarop\left(#1\right)}
\newcommand{\TD}[1]{\mathrm{TD}\/\left[#1\right]}
\newcommand{\Sunique}{\Sigma_\mathrm{unique}}
\newcommand{\Irr}[1]{\Irrop\left[#1\right]}

\newtheoremstyle{thrmstyle}
  {0pt} % Space above
  {0pt} % Space below
  {\em} % Body font
  {} % Indent amount
  {\bfseries} % Theorem head font
  {.} % Punctuation after theorem head
  {.5em} % Space after theorem head
  {} % Theorem head spec
\theoremstyle{thrmstyle}
\newtheorem{theorem}{Theorem}
\newtheorem{corollary}{Corollary}
\newtheorem{lemma}{Lemma}
\newtheorem{fact}{Fact}

\newtheorem{definition}{Definition}

\newtheoremstyle{definition_new}
  {10pt} % Space above
  {5pt} % Space below
  {} % Body font
  {} % Indent amount
  {\bfseries} % Theorem head font
  {.} % Punctuation after theorem head
  {.5em} % Space after theorem head
  {} % Theorem head spec
\theoremstyle{definition_new}

\newtheorem*{observation}{Observation}

\newtheorem*{remark}{Remark}
\newtheorem{step}{Step}

\renewenvironment{proof}{\emph{\bfseries Proof.}}{\qed\vspace{5pt}}

\ytableausetup{boxsize=0.5cm, centertableaux}
\newlength{\gridboxsize}
\usepackage{titlesec}
\titleformat*{\section}{\large\bfseries}
\titleformat*{\subsection}{\normalsize\bfseries}
\titleformat*{\subsubsection}{\normalsize\bfseries}
\titleformat*{\paragraph}{\normalsize\bfseries}
\titleformat*{\subparagraph}{\normalsize\bfseries}

\usepackage[left=1in,right=1in,top=1in,bottom=1in]{geometry}
\usepackage[skip=10pt, indent=0pt]{parskip}
 \title{\vspace{-20pt}\large\bf Pseudorandomness from Subset States}
 \author[]{\normalsize Tudor Giurgic\u{a}-Tiron\thanks{\href{mailto:tgt@stanford.edu}{tgt@stanford.edu}} }
 
 \author[]{\normalsize Adam Bouland\thanks{\href{mailto:abouland@stanford.edu}{abouland@stanford.edu}}}

 \affil{\vspace{-5pt}\small Stanford University}
\date{}

\begin{document}

\maketitle
\vspace{-10pt}

\begin{abstract}\noindent We show it is possible to obtain quantum pseudorandomness and pseudoentanglement from random subset states --- i.e. quantum states which are equal superpositions over (pseudo)random subsets of strings. 
This answers an open question of Aaronson et al. \cite{aaronson2022quantum}, who devised a similar construction augmented by pseudorandom phases. 
Our result follows from a direct calculation of the trace distance between $t$ copies of random subset states and the Haar measure, via the representation theory of the symmetric group. We show that the trace distance is negligibly small, as long as the subsets are of an appropriate size which is neither too big nor too small. 
In particular, we analyze the action of basis permutations on the symmetric subspace, and show that the largest component is described by the Johnson scheme: the double-cosets of the symmetric group $\bbS_N$ by the subgroup $\bbS_t \times \bbS_{N-t}$. 
The Gelfand pair property of this setting implies that the subset state density matrix eigenbasis coincides with the symmetric group irreducible blocks, with the largest eigenblock asymptotically approaching the Haar average. 
An immediate corollary of our result is that quantum pseudorandom and pseudoentangled state ensembles do not require relative phases.
\end{abstract}

\vspace{10pt}

\section{Introduction}

Subset states are a natural quantum combinatorial construction, representing uniform superpositions over the subsets of an orthonormal basis.
\begin{definition}
The {\em subset state} $\ket{S}$ associated with the subset $S\subseteq[N]$ is defined by:
    \begin{equation}\label{eq:subsetstatedef}
    \ket{S} \equiv \frac{1}{\sqrt{\abs{S}}}\displaystyle\sum_{x\in S} \ket{x}\,,
    \end{equation}
\end{definition}
where $\{\ket{x}\}_{x\in[N]}$ denotes the computational basis of an $n$-qubit system, such that $[N]=\B^n$.
Such states have found numerous applications, for example in $\QMA$ protocols for group non-membership \cite{watrous2000succinct}, as general $\QMA$ witness-approximators \cite{grilo2015qma}, in building oracle separations between $\QMA$ and $\QCMA$ \cite{fefferman2015quantum}, and in proposals for quantum money protocols \cite{aaronson2012quantum}. 

Recently a variant of subset states arose in the study of quantum pseudorandomness. In particular, Aaronson et al. \cite{aaronson2022quantum} introduced a variation known as the subset-phase state:

\begin{definition}
A {\em subset-phase state} is a state specified by both a subset $S\subseteq [N]$ and a Boolean function $f:S\rightarrow \B$, defined as:
\begin{equation}
\ket{S, f} \equiv \frac{1}{\sqrt{|S|}}\displaystyle\sum_{x\in S} (-1)^{f(x)}\ket{x}\,.
\end{equation}  
\end{definition}

That is, a subset-phase state is simply a subset state which has been augmented with relative phases between the subset elements in the superposition. The main result of
Aaronson et al. \cite{aaronson2022quantum} is that pseudorandom subset-phase states of superpolynomial subset size\footnote{I.e., 
if the set $S$ is a pseudorandom subset of superpolynomial size, and the function $f$ is a pseudorandom function (PRF).} are a pseudorandom state ensemble.
In other words, such subset-phase states are both efficiently preparable and computationally indistinguishable from the Haar measure, even given access to an arbitrary polynomial number of copies of the state \cite{ji2018pseudorandom}.
This construction allowed them to derive a number of results.
For example, since the subset size is an upper bound on the Schmidt rank across arbitrary cuts, this construction showed that it is possible to have pseudorandom state ensembles with very low entanglement across every cut of the system simultaneously, saturating the prior lower bound \cite{ji2018pseudorandom}.
It also led to the definition of pseudoentangled state ensembles \cite{aaronson2022quantum}, which are ensembles of states with differing entanglement structures which are nevertheless computationally indistinguishable.

A natural question, which was posed as an open problem in \cite{aaronson2022quantum}, is whether or not the phases can be removed from the construction. That is, do pseudorandom subset states already form a pseudorandom ensemble? This is far from obvious because phases have provided the main mechanism for constructing pseudorandom states in the literature, going back to the original construction of Ji, Liu, and Song \cite{ji2018pseudorandom}, and to the simplification due to Brakerski and Shmueli \cite{brakerski2019pseudo}.
The core result of these prior works was an information-theoretic statement: the Haar measure can be approximated by truly random binary phase states --- namely, the special case of subset-phase states $\ket{S, f}$ for which the subset is the full set $S=[N]$.
Pseudorandomness is obtained from this construction by a simple hybrid argument, substituting pseudorandom phases for the truly random phases.
This prior result was used in a critical way in Aaronson et al.'s proof, which proceeded by a similar information-theoretic calculation, in two steps. 
First, the phase state result can be applied at the level of each subset $S \in [N]$, such that averaging over the phase function $f$ provides proximity to the Haar measure restricted to a single subset.
Second, Aaronson et al. show that subsequently performing the average over subsets provides enough uniform coverage of the symmetric subspace to approximate the Haar measure. Without random phases, the first step of this argument does not work, so new techniques are needed to approach this setting.

\subsection{Our result}

In this work, we answer this open problem in the affirmative. Our main result is that under a uniform distribution over subsets of size $m$, random subset states form approximate designs, as long as the subsets are of an appropriate size. Formally, we prove:

\begin{theorem}\label{thrm:randomsubsetstates}
        Let $t=O(\poly(n))$, let $m$ be in the range $\omega(\poly(n)) < m < o\left(2^n\right)$, and let $N=2^n$. Let $\TD{\cdot, \cdot}$ denote the trace distance between two matrices. Then we have that:
        \begin{equation}\label{eq:tracedistance}
            \TD{\E_{S \in \binom{[N]}{m}}\dyad{S}^{\otimes t},\;\E_{\ket\phi\sim\Haar{[N]}}\dyad{\phi}^{\otimes t}}\leq O\left(\frac{tm}{N}\right) + O\left(\frac{t^2}{m}\right)\,.
        \end{equation}
\end{theorem}

In other words, the trace distance between a random subset state and the Haar measure is negligible for arbitrary polynomial number of copies $t$, as long as the subset size $m$ is in a certain ``Goldilocks'' regime --- not too small (superpolynomial in the number of qubits $n$), but also not too big (superpolynomially smaller than the maximum value $2^n$).

The two terms bounding the trace distance in \eqref{eq:tracedistance} have natural operational interpretations in terms of the success rate of possible distinguishers. The first term describes the `birthday attack': 
a natural way to distinguish subset states from Haar random states is to measure the $t$ copies of the state in the computational basis, and see if one obtains the same outcome more than once.
This occurs with probability $O(t^2/m)$, so if $m$ is small enough relative to $t^2$ there would be a non-negligible probability of this attack succeeding. 
The second term is related to a projective measurement against the $\ket{+^n}=\ket{[N]}$ state --- if the subset size $m$ is too large compared to $N$, the measurement success probability $\abs{\ip{S}{+^n}}^2=m/N$ will be high, which means that the subset state cannot have come from the Haar measure. 

Interestingly, our result implies that the Goldilocks regime for the subset size is as large as possible, given these two simple distinguishing algorithms.
For example, one corollary of our result is that one cannot distinguish the case of subset states with subset size $m=2^{\log^2 n}$ from the case $m=2^{n-\log^2 n}$.

Additionally, the information-theoretic result of theorem \ref{thrm:randomsubsetstates} immediately implies the following corollary:
\begin{corollary}

    Pseudorandom subset states are a pseudorandom and pseudoentangled state ensemble so long as the size of the subsets is superpolynomial in $n$ and superpolynomially less than $2^n$.
    
\end{corollary}

This follows from a direct hybrid argument as in \cite{ji2018pseudorandom} --- as random subset states of the appropriate size are information-theoretically indistinguishable from Haar, pseudorandom subset states are automatically computationally indistinguishable from Haar. Pseudorandom subset states admit efficient preparation by the arguments given in \cite{aaronson2022quantum}, namely one can prepare a fixed subset by Hadamard'ing a subset of qubits, and then applying a quantum-secure pseudorandom permutation (PRP) \cite{zhandry2016note} in place. We note that this is possible since the PRP is invertible given the secret key.

An interesting direct corollary of our result is that quantum pseudorandom states can not only be real \cite{brakerski2019pseudo}, but they also do not even require negative numbers in the state vectors.
This stands in sharp contrast to pseudorandom unitaries, which must have a large imaginary component \cite{haug2023pseudorandom}.

We also achieve a slight generalization of our result which unifies our main theorem and Aaronson et al.'s trace distance calculation for subset-phase states.
In particular, we consider the ensemble of subset-phase states with \emph{biased} phases, i.e. where the phase on each element is $+1$ with probability $\frac{1+\delta}{2}$ and $-1$ with probability $\frac{1-\delta}{2}$.
Subset states correspond to $\delta=1$, and the subset-phase states of Aaronson et al. correspond to $\delta=0$. We obtain an interpolation between the two trace distance bounds as a function of $\delta$, and show that one needs near-maximal `sign entropy' to modify the dependency on the subset size $m$ --- see section \ref{sec:biasedphases} for details. 

Finally, we note that we recently became aware of independent and concurrent work by Fermi Ma proving a similar result \cite{macommunication}.

\subsection{Proof strategy}

Before we lay out the technical details behind our result in section \ref{sec:subsetstateproof}, it is worth providing a brief high-level overview of our approach. The average over $t$-copies of a Haar-random state is proportional to the projector onto the symmetric subspace $\Sym{t}{[N]}$ (fact \ref{fact:haarsymproj} below), and the single-subset state contributions $\dyad{S}^{\otimes t}$ lie in the symmetric space by construction. Therefore it is enough to show closeness to the maximally mixed state from within the symmetric subspace itself.

The high-level motivation is to realize that uniform average over subset states is an object with a high degree of symmetry, and to take advantage of this symmetry in order to identify a large subspace within the ambient symmetric subspace on which the density matrix diagonalizes in a tractable fashion. Informally, we would like to find a `typical subspace' which satisfies the conditions:
\begin{enumerate}
    \item The typical subspace occupies most of the dimensions of the symmetric subspace $\Sym{t}{[N]}$.
    \item It admits a natural group action which is compatible with the symmetries of the subset state average.
    \item The structure of the subspace under the group action is particularly simple, being described by a single irreducible representation. This reduces the problem of computing the restricted trace distance to evaluating one eigenvalue of high multiplicity.
\end{enumerate}
We show it is possible to find such a subspace. We also show that these properties are enough to bound the total trace distance --- this is shown as an auxiliary result (lemma \ref{lemma:nearbymatrices} in section \ref{sec:helpfullemma}). Our irreducible subspace lies at the end of a two-step restriction, in which every stage maintains the vast majority of dimensions:
\begin{equation*}
\begin{array}{cclclc}
     &V_{[N-t, t]} &\subset& \Sunique &\subset& \Sym{t}{[N]} \\
     &{\scriptstyle \text{(largest } \bbS_N\text{-irrep block)}} & & {\scriptstyle\text{(unique type subspace)}} & & {\scriptstyle \text{(full symmetric subspace)}}\,. \\
     {\scriptstyle \text{dimension:}} & \binom{N}{t}-\binom{N}{t-1} &{\scriptstyle \approx_{O\left(t/N\right)}}& \binom{N}{t} &{\scriptstyle \approx_{O\left(t^2/N\right)}}& \binom{N+t-1}{t}
\end{array}
\end{equation*}

The first restriction is from the full symmetric subspace $\Sym{t}{[N]}$ to $\Sunique$, the subspace of unique types (defined in section \ref{sec:symspacebackground}). In the limit of $t \ll N$, only a fraction of $O(t^2/N)$ dimensions are lost under this restriction, due to the `birthday problem' asymptotics on the type basis of the symmetric subspace. When restricting the subset states themselves to this subspace, we incur a similar $O(t^2/m)$ combinatorial correction, which is the origin of the first trace distance term in theorem \ref{thrm:randomsubsetstates}. The second restriction is from the unique-type subspace $\Sunique$ to a subspace $V_{[N-t,t]}$, which arises as the largest irreducible representation block when decomposing $\Sunique$ under permutation action. This restriction is similarly typical, maintaining all but a subleading fraction $O(t/N)$ of dimensions. Overall, the largest irrep subspace $V_{[N-t,t]}$ occupies most of the total symmetric subspace, and will serve as our choice of typical subspace, on which we will be able to show proximity to the Haar average.

The subspace of unique types $\Sunique$ is indexed by subsets of size $t$ --- in this basis, the density matrix entries decay away from identity at the rate of $m/N$ as a function of the Hamming distance between subsets. To derive the associated spectral properties, we need to deal with the geometry of subset intersections, which we approach from the point of view of permutation action. The key technique is the explicit description of the subspace $\Sunique$ under the natural action of the symmetric group $\bbS_N$, i.e. permuting the computational basis. Formally, subsets of size $t$ can be seen as the cosets of $\bbS_N$ by the subset-preserving subgroup $\bbS_t \times \bbS_{N-t}$. This particular homogeneous space (the so-called {\em Johnson scheme}) has been extensively studied before. It is known that this space admits a particularly simple multiplicity-free irreducible decomposition with $t+1$ terms, owing to the fact that it obeys the Gelfand pair property. Conveniently, the uniform mixture of subset states is invariant under such permutations of the basis. As a consequence, the diagonal basis of the density matrix and the $\bbS_N$-irrep basis coincide on this subspace, and we can show that the largest irrep $V_{[N-t,t]}$ takes up the vast majority of the space. The theory of Gelfand pairs also provides the relevant spherical functions by which we analytically evaluate the most frequent eigenvalue associated with this eigenblock, which we show is only a relative correction of $O(tm/N)$ away from the Haar average. This is the origin of the second trace distance term in theorem \ref{thrm:randomsubsetstates}.

\section{Random subset states via non-Abelian harmonic analysis}
\label{sec:subsetstateproof}

In this section we explain the technical proof of theorem \ref{thrm:randomsubsetstates}, as well as a brief summary of the required background.  This section is divided into four parts. First, we start by establishing the core facts about the symmetric subspace and `birthday problem' typicality. Second, we outline the representation-theoretic background involving homogeneous spaces which respect the Gelfand pair property. In the third section, we apply these tools to prove theorem \ref{thrm:randomsubsetstates}, by diagonalizing (most of) the density matrix and computing its typical eigenvalue and showing it approaches the Haar average. In the fourth and final section, we outline a technical lemma which allows us to bound the total trace distance by only using information about a large enough subspace. Throughout, we will assume familiarity with the basic tools of representation theory.

\subsection{Background on the symmetric subspace}\label{sec:symspacebackground}

    Let us first review some well-known facts about the combinatorics of symmetric subspaces and establish the relevant notation for our task. For a broader review of similar techniques, see for example \cite{harrow2013church}.

    \begin{definition}
        Given a subset $S \subseteq [N]$, we define the $(S, t)$-symmetric subspace $\Sym{t}{S}$ as the symmetric subspace of $t$ copies of the space spanned by the sub-basis $\{\ket{j}\}_{j\in S}$. The full symmetric subspace $\Sym{t}{\C^N}$ is then denoted $\Sym{t}{[N]}$.
    \end{definition}
    
    \begin{fact}[Birthday asymptotics]\label{fact:birthdayasymptotics}
    The `birthday problem' can be summarized in the form of the small-$k$ relative correction to the ratio:
    \begin{align}
        \frac{n!}{(n-k)!} &= n^k\left(1+O\left(\frac{k^2}{n}\right)\right)\,.
    \end{align}
    Let us also record an immediate corollary:
    \begin{align}
        \binom{n\pm\l}{k} &= \frac{n^k}{k!}\left(1+O\left(\frac{k\l}{n}\right)+O\left(\frac{k^2}{n}\right)\right)\,.
    \end{align}
    \end{fact}
    
    \begin{fact}[Haar-averaging projects onto the symmetric subspace]\label{fact:haarsymproj}
        The average of $t$ copies of a Haar-random state is the maximally mixed state over the symmetric subspace $\Sym{t}{[N]}$:
        \begin{equation}
            \E_{\phi \sim \Haar{[N]}} \dyad{\phi}^{\otimes t} = \frac{1}{\binom{N+t-1}{t}}\Pi_{\Sym{t}{[N]}}\,.
        \end{equation}
    \end{fact}
    \begin{fact}[The type basis for the symmetric subspace]\label{fact:typebasis}
        The canonical orthonormal type basis for $\Sym{t}{S}$ is composed of the single-type vectors $\ket{\theta}$, for types identified by size-$t$ multisets $\theta \in \MSet(S, t)$. Specifically, a type/multiset $\theta$ is an unordered collection of $t$ elements from $S$, with repetitions allowed. Therefore, if $S=\{x_1, \dots, x_m\}$, then a type $\theta \in \MSet(S,t)$ can be equivalently defined as a sequence of $m$ non-negative integers $(\theta_{x_1}, \theta_{x_2}, \dots, \theta_{x_m})$ with $\theta_{x_1}+\dots+\theta_{x_m}=t$, representing the number of occurences of each $x_k$. For such a $\theta$ type, the corresponding basis vector $\ket{\theta}$ is proportional to all $t$-permutations of basis vectors described by this type, namely:
    \begin{align}
    \ket{\theta}&=\sqrt\frac{\abs{\theta}}{t!} \sum_{\sigma \in \bbS_t}\sigma\ket{\underbrace{x_1,\dots,x_1}_{\theta_{x_1}\text{ times}},\;\underbrace{x_2,\dots,x_2}_{\theta_{x_2}\text{ times}},\dots\dots,\underbrace{x_m,\dots,x_m}_{\theta_{x_m}\text{ times}}} & \text{ where }\abs{\theta}\equiv \theta_{x_1}!\dots \theta_{x_m}!\,.
    \end{align}
    Here, the action is the usual site permutation, namely $\sigma\ket{x_1,\dots, x_t}=\ket{x_{\sigma^{-1}(1)},\dots, x_{\sigma^{-1}(t)}}$. There are $\dim\Sym{t}{S}=\abs{\MSet(S,t)}=\binom{\abs{S}+t-1}{t}$ distinct types.
    \end{fact}

    \begin{fact}[Most of the symmetric subspace is spanned by unique types]\label{fact:birthdaytypes}
        If all frequencies $\alpha_x$ in a type $\alpha \in \MSet(S,t)$ are either $0$ or $1$ (i.e. no duplicates), then we refer to $\alpha$ as a unique type. This corresponds to restricting the multisets to conventional sets, such that each unique type corresponds to a $t$-subset of $S$, which we will simply denote  by $\alpha \in \binom{S}{t}$. Let $\Sunique=\mathrm{span}\{\ket{\alpha}\}_{\alpha \in \binom{[N]}{t}}$ be the subspace of unique types inside the full symmetric subspace $\Sym{t}{[N]}$. The birthday asymptotics tell us that when $t$ is much smaller than $N$, most of the symmetric subspace is occupied by the unique subspace $\Sunique$, up to a small relative fraction of $O(t^2/N)$ dimensions, since:
        \begin{align}
            \underbrace{\dim \Sunique}_{=\binom{N}{t}} = \underbrace{\dim \Sym{t}{[N]}}_{=\binom{N+t-1}{t}}\left(1 + O\left(\frac{t^2}{N}\right)\right)\,.
        \end{align}
    \end{fact}

\begin{remark}
    A word of warning before we proceed: the unique-type states $\ket{\alpha}$ which form an orthonormal basis for $\Sunique$ are {\em not} subset states of the form \eqref{eq:subsetstatedef}! This is despite the fact that they are defined by subsets of $[N]$. For subset $S\in\binom{[N]}{m}$, the subset state $\ket{S}$ lies in the original Hilbert space $\C^N$, while for subset $\alpha \in \binom{[N]}{t}$, the unique-type state $\ket{\alpha}$ lies in the symmetric subspace within the $t$-copied space $(\C^N)^{\otimes{t}}$. We will maintain the notation of uppercase-Latin letters for subset states (e.g. $\ket{S}$), and lowercase-Greek letters for type states (e.g. $\ket{\alpha}$), in order to emphasize this difference.
\end{remark}

\subsection{Background on Gelfand pairs and homogeneous spaces}\label{sec:RTbackground}

    Here, we outline the basic facts about group theory, homogeneous spaces, finite Gelfand pairs, and the specific case of $\bbS_N \,/\, \bbS_t \times \bbS_{N-t}$ relevant to our problem. Most of these facts are condensed from the relevant literature, in particular we point the interested reader to section 3F in the book by Diaconis \cite{diaconis1988group}, and to chapter 6 in the book by Ceccherini-Silberstein, Scarabotti, and Tolli \cite{ceccherini2008}.

        {\bf Preliminaries. } Let $G$ be a finite group, $K \leq G$ a subgroup (called the isometry subgroup), and $X= G/K$ the homogeneous space of right-$K$-cosets of the form $Kg$. The associated action of $G$ on $X$ is $g:\,x\mapsto xg^{-1}$. Denote by $\Irr{G}$ the set of all irreducible representations (irreps) of $G$, such that $\rho_\lambda:G \to \text{U}(d_\lambda)$ is a unitary $G$-irrep in dimension $d_\lambda$, for each $\lambda \in \Irr{G}$. Furthermore, denote $\Pi_K^\lambda\equiv \frac{1}{\abs{K}}\sum_{k\in K}\rho_\lambda(k)$ the projector onto the subspace fixed by $K$ inside the $\lambda$ irrep.

        To ease both notation and intuition, we will not make a notation distinction between the elements of the homogeneous space $X$ and a choice of $K$-coset representatives seen as elements of $G$. We will denote by $\Irr{L(X)}$ the set of irreps of $G$ which have nonzero multiplicity in the decomposition of $L(X)=\C^X$ (the set of complex-valued functions defined on $X$, alternatively denoted $\mathbb C[X]$, also called the permutation representation).

    The notion of {\em Gelfand pair} is a special relationship between a group $G$ and its subgroup $K$ with particularly simple representation-theoretic properties. In our application, we will see that the space of unique types $\Sunique$ can be described by a canonical choice of Gelfand pair.

    \begin{definition}[Finite Gelfand pairs]
        For finite group $G$ and subgroup $K < G$, the following are equivalent definitions of $(G, K)$ forming a Gelfand pair:
        \begin{enumerate}
            \item[(1)] The algebra of bi-$K$-invariant functions $f(k_1gk_2)=f(g)$ under $G$-convolution is commutative (the textbook definition).
            \item[(2)] For $X=G/K$, the decomposition of $L(X)$ into $G$-irreps is multiplicity-free.
            \item[(3)] In every irreducible $G$-representation $\lambda$ present in the decomposition of $L(X)$, the subspace fixed by $K$ is one-dimensional. In other words, the projector $\Pi_K^\lambda$ is of rank one, and we denote $\Pi_K^\lambda=\dyad{k^\lambda}$.
        \end{enumerate}
    \end{definition}

    Homogeneous spaces over Gelfand pairs are common in applications involving group symmetry across statistics and combinatorics, for example in studying the convergence properties of random walks over groups \cite{diaconis1988group}. In quantum information, Gelfand pairs over the unitary group $\text{U}(d)$ have recently been used to recursively construct exact unitary $t$-designs for arbitary $t$ and $d$ \cite{bannai2019explicit, nakata2021quantum}.
    
    It is well-known that the conventional Abelian Fourier transform over $\Z_N$ diagonalizes a circulant matrix --- i.e. a matrix for which the $(i,j)$ entry only depends on the difference $i-j$. We will use a generalization of this notion to the non-Abelian setting, which on a homogeneous space is easily shown to be identical to group invariance:

    \begin{fact}[Group-invariant matrices are group-circulant]
        For $X=G/K$, a matrix $\calM \in \C^{X \times X}$ commutes with the group action of $G$ if and only if $\calM$ is $G$-circulant, i.e. there exists a function $\nu : K\backslash G/K \to \C$ such that:
        \begin{equation}
            \calM_{xy} = \nu(x y^{-1})\,.
        \end{equation}
        In other words, the entries depend only on the `distance from diagonal' in terms of the group operation on $G$. Note that a well-defined circulant function $\nu$ is constant on the double cosets of $K$, i.e. $\nu(\gamma)=\nu(k_1\gamma k_2),\;\;\forall \gamma\in G,\;k_{1,2}\in K$.
    \end{fact}

    \begin{proof}
        Under group action, the matrix transforms as $\left(g\calM g^{-1}\right)_{x,y}=\calM_{xg, yg}$. The `if' direction is immediate since $\nu((xg)(yg)^{-1})=\nu(xy^{-1})$. In the `only if' direction, invariance of the matrix under group action means that $\calM_{xg, yg}=\calM_{x,y}$ for all $x,y\in X$ and $g\in G$. Picking $g$ from the left coset $y^{-1}K$ means that $\calM_{x,y}=\calM_{xy^{-1},\text{id}}$, for all $x,y\in X$. Then the circulant condition is satisfied, with circulant function $\nu(x)=\calM_{x,\text{id}}$.
    \end{proof}
    
    In analogy to the Abelian case, the non-Abelian Fourier transform over the group $G$, which changes the basis into the irrep-block diagonal basis, achieves a partial diagonalization of a commutant matrix. When imposing the additional Gelfand pair property, we will see that the matrices are completely diagonalized by this basis change:

    \begin{fact}[Diagonalization of circulant matrices]\label{lemma:diagonalization}
        For Gelfand pair $(G, K)$ and $X=G/K$, there is a unitary matrix $W \in \C^{X\times X}$ (the $G$-Fourier transform), such that any matrix $\calM \in \C^{X\times X}$ which commutes with the $G$-action on $X$ is diagonalized by $W$:
        \begin{equation}
            W\calM W^\dagger = \bigoplus_{\lambda \in \Irr{L(X)}} \tilde\mu_\lambda \cdot I_{d_\lambda}\,.
        \end{equation}
        If the matrix circulant function is $\calM_{xy}=\nu(xy^{-1})$, the eigenvalues are given by:
        \begin{align}\label{eq:eigenvalues}
            \tilde{\mu}_\lambda &= \sum_{x\in X} \nu(x) \Phi^{(G,K)}_{\lambda}(x)\,,
        \end{align}
        where $\Phi^{(G,K)}_\lambda : X \to \C$ is known as the {\em spherical function}:
    \begin{equation}
        \Phi^{(G,K)}_{\lambda}(x) \equiv  \tr[\rho_\lambda(x)\Pi_K^\lambda]= \mel{k^\lambda}{\rho_\lambda(x)}{k^\lambda}\,.
    \end{equation}
    \end{fact}

    \begin{proof}
        This is a standard application of Schur's lemma. Let the space $L(X)$ have a generic $G$-irrep decomposition into irreducible $G$-modules with nonzero multiplicities $m_\lambda$. $W$ is the unitary basis change on $L(X)$ associated with this view. In this basis, the action of $G$ on $L(X)$ has the matrix structure:
        \begin{equation}
            WR(g)W^\dagger = \bigoplus_{\lambda \in \Irr{L(X)}} \rho_\lambda(g) \otimes I_{m_\lambda}\,.
        \end{equation}
        By Schur orthogonality, a matrix $\calM$ acting on $L(X)$ which commutes with the $G$-action is irrep-block diagonal:
        \begin{equation}
            W\calM W^\dagger = \bigoplus_{\lambda \in \Irr{L(X)}} I_{d_\lambda} \otimes Q_\lambda\,,
        \end{equation}
        for some matrices $Q_\lambda$ of size $m_\lambda \times m_\lambda$.

        For a Gelfand pair, the corresponding irrep-decomposition is guaranteed to be multiplicity-free, i.e. $m_\lambda=1$, which means that in this basis the $Q_\lambda$ matrices are simply scalars, i.e. the eigenvalues $\tilde\mu_\lambda$.
        
        Recall that for the $G$-action $R$ acting on $L(X)$, the projector onto the subspace corresponding to irrep $\lambda \in \Irr{G}$ is given in terms of the irreducible $G$-characters $\chi_\lambda$ by $P_\lambda^R \equiv \frac{d_\lambda}{\abs{G}}\sum_{g\in G}\chi_\lambda(g)R(g)$. This allows us to express the eigenvalues in terms of the trace over the $\lambda$-block:
        \begin{align}
            \tilde\mu_\lambda &= \frac{1}{d_\lambda} \Tr[P_\lambda^R\,\calM]&\\
            &= \frac{1}{\abs{G}} \sum_{g\in G} \chi_\lambda(g) \sum_{x\in X}\calM_{xg, x}&\\
            &= \frac{1}{\abs{K}} \sum_{g \in G} \chi_\lambda(g) \nu(g)& \text{(for circulant matrix function $\nu$)}\\
            &= \sum_{x \in X} \nu(x)\left(\frac{1}{\abs{K}}\sum_{k\in K}\chi_\lambda(kx)\right).&\text{(circulant function $\nu$ is $K$-invariant)}\\
            &= \sum_{x\in X} \nu(x) \tr[\rho_\lambda(x)\Pi_K^\lambda]\,.
        \end{align}
    \end{proof}
    \vspace{-20pt}
    \begin{remark}
        A standard representation-theoretic fact is that the irreducible composition of $L(G/K)$ is the same as that of $\mathrm{Ind}_{K}^{G}1$, i.e. inducing the trivial representation on $K$ up to $G$.
    \end{remark}

    In our application, the density matrix arising from averaging over subset states will exhibit precisely this kind of circulant property, thus admitting diagonalization along the irrep blocks. A key benefit of the theory of Gelfand pairs is being able to work in the algebra spanned by the spherical functions directly, which are generally much simpler objects than irreducible characters, and can often be obtained analytically. This is precisely the case in our application: below, we detail the relevant facts which apply to our particular setting of interest on the symmetric group, by collecting a standard series of results from \cite{diaconis1988group, ceccherini2008}.

    \begin{fact}[The Johnson scheme]\label{fact:johnsonfacts}
        Consider the case $G=\bbS_N$, with the isometry subgroup $K=\bbS_t \times \bbS_{N-t}$. Then we have that:
        \begin{enumerate}
            \item The cosets $X=\bbS_N\,/\,\bbS_t \times \bbS_{N-t}$ correspond to the set $\binom{[N]}{t}$ of subsets of $[N]$ of size $t$. For a $g \in \bbS_N$, the standard mapping is to the subset $\alpha=g([t])$, the image of $[t]$ under $g$.
            \item $(\bbS_N,\;\bbS_t \times \bbS_{N-t})$ is a Gelfand pair.
            \item Recall that the irreps of $\bbS_N$ are indexed by Young diagrams of size $N$, which correspond to integer partitions of $N$. There are $t+1$ different $\bbS_N$-irreps which appear in the decomposition of $L(X)$, specifically:
            \begin{align}
                L(X) &\simeq \bigoplus_{q=0}^t V_{[N-q, q]}\,.
            \end{align}
            In other words, the relevant irreps are indexed by the Young diagrams with one or two rows $[N], [N-1,1], \dots, [N-t,t]$.
            \item There are $t+1$ orbits of $X$ under right-action by $K$ (corresponding to the double-cosets $K\backslash G/K$), indexed by the subset distance $p\in\{0,1,\dots,t\}$. The orbit of $X$ associated with subset distance $p$ has $\binom{t}{p}\binom{N-t}{p}$ elements, which is the number of $t$-subsets at Hamming distance $p$ from a fixed subset.
            \item For the irrep $\lambda = [N-q, q]$, the spherical function is known to have the analytical form, as a function of the distance $p\in\{0,\dots,t\}$:
            \begin{equation}\label{eq:sphericalfunction}
                \Phi_{[N-q,q]}(p) = \sum_{k=0}^q (-1)^k \frac{\binom{q}{k}\binom{p}{k}\binom{N-q+1}{k}}{\binom{t}{k}\binom{N-t}{k}}\,\quad\quad\text{ for }\; 0\leq p \leq t\,.
            \end{equation}
        \end{enumerate}
    \end{fact}

    Now we have all the necessary tools to proceed with the proof of our main result in the next section.
    
    \subsection{Proof of theorem \ref{thrm:randomsubsetstates}}

    Let our density matrix obtained by subset state-averaging be denoted by:
    \begin{equation}
        \rho \equiv \E_{S \sim \binom{[N]}{m}} \dyad{S}^{\otimes t}\,.
    \end{equation}
    Also denote the maximally mixed state on the symmetric subspace, which comes from Haar integration, as:
    \begin{align}
        \rho_0 &\equiv \E_{\ket{\phi}\sim\Haar{[N]}} \dyad{\phi}^{\otimes t}\\
         &= \frac{1}{\binom{N+t-1}{t}}\Pi_{\Sym{t}{[N]}}\,.
    \end{align}
    Our goal is to upper-bound the trace distance $\TD{\rho, \rho_0}$. By design, $\rho$ is included in the support of the symmetric subspace, so we only need to care about diagonalizing this matrix inside the symmetric subspace. In fact, we will see that is enough to diagonalize this matrix inside a convenient choice of a smaller subspace inside $\Sym{t}{[N]}$.

    \begin{step}[Density matrix entries in the $\Sunique$ subspace]
        The density matrix $\rho$ is an average over single-subset contributions of the form $\dyad{S}^{\otimes t}$. Let us express this single-subset term in the type basis outlined in fact \ref{fact:typebasis}:
    \begin{align}
        \dyad{S}^{\otimes t} &= \frac{1}{\abs{S}^t}\sum_{\substack{x_1,\dots,x_t \in S\\y_1,\dots,y_t\in S}}\dyad{x_1,x_2,\dots,x_t}{y_1,y_2,\dots,y_t}\\
        &= \frac{t!}{\abs{S}^t}\sum_{\theta,\varphi \in \MSet(S,t)} \frac{1}{\sqrt{\abs{\theta}\abs{\varphi}}}\dyad{\theta}{\varphi}\,.\label{eq:subsetstatetypesum}
    \end{align}

    Let us now restrict our analysis to the subspace of unique types, $\Sunique$, spanned by the unique-type states $\ket{\alpha}$ associated with single subsets $\alpha \in \binom{[N]}{t}$. Then, starting from \eqref{eq:subsetstatetypesum}, the matrix entries in this subspace are:
    \begin{align}\label{eq:restrictedrhoentries}
        \mel{\alpha}{\rho}{\beta} &= \E_{S\sim\binom{[N]}{m}} \frac{t!}{\abs{S}^t}\;\mathop{\mathbf{1}\/}\left[\alpha\cup\beta\subset S\right] & \text{for } \alpha,\beta \in \binom{[N]}{t}\\
        &= \frac{t!}{m^t}\frac{1}{\binom{N}{m}}\abs{\left\{\left.S \in \binom{[N]}{m}\;\right\vert\;\alpha\cup\beta \subset S\right\}}\,. &
    \end{align}
    To count how many subsets of size $m$ contain $\alpha\cup\beta$, we can equivalently count the ways in which we can append elements to $\alpha\cup\beta$ from among the other $[N]\setminus(\alpha\cup\beta)$ elements until we get a set of size $m$, which can be done in $\binom{N - \abs{\alpha\cup\beta}}{m - \abs{\alpha \cup \beta}}$ ways. This gives us the exact matrix entries in this subspace:
    \begin{align}
        \mel{\alpha}{\rho}{\beta} &= \frac{t!}{m^t}\frac{\binom{N - \abs{\alpha\cup\beta}}{m - \abs{\alpha \cup \beta}}}{
        \binom{N}{m}
        }\,.\label{eq:exactMentries}
    \end{align}
    Applying the birthday asymptotics from fact \ref{fact:birthdayasymptotics}, we further refine:
    \begin{align}
        \mel{\alpha}{\rho}{\beta} &= \frac{1}{\binom{N+t-1}{t}}\left(\frac{m}{N}\right)^{\abs{\alpha\cup\beta} - t}\left(1 + O\left(\frac{t^2}{m}\right)\right)\,.\label{eq:correctiontoverm}
    \end{align}
    Note that relative corrections of size $O(t^2/N)$ also appear, but we do not record them since they are subleading with respect to $O(t^2/m)$. Define the matrix $\calM$ as the rescaled density matrix, restricted to the subspace $\Sunique$:
    \begin{equation}
        \calM \equiv \binom{N+t-1}{t}\;\Pi_{\text{unique}}\;\rho\;\Pi_{\text{unique}}\;\;\in\C^{\binom{N}{t}\times\binom{N}{t}}\,,
    \end{equation}
    where the projector onto $\Sunique$ is simply $\Pi_\text{unique}=\sum_{\alpha\in\binom{[N]}{t}}\dyad{\alpha}$. We will ultimately show that the eigenvalues of the matrix $\calM$ are negligibly close to one, for the vast majority of eigenvalues.
    \end{step}

    \begin{observation}
        Naturally, $\bbS_N$ acts by basis permutation on the full symmetric subspace $\Sym{t}{[N]}$, with a structure given by the recently-defined {\em multiset partition algebra} \cite{narayanan2023multiset, orellana2023howe}. By restricting to the dominant unique-type subspace $\Sunique$, we can instead work with a much simpler action of $\bbS_N$ described by the tools of section \ref{sec:RTbackground}. Due to the bijection between subsets $\binom{[N]}{t}$ and the right-cosets of $G=\bbS_N$ by $K=\bbS_{t}\times \bbS_{N-t}$, we identify $\Sunique$ with the homogeneous space $L(X)$, where $X=G/K$, for which the Johnson scheme (fact \ref{fact:johnsonfacts}) applies.
    \end{observation}

    \begin{step}[The matrix $\calM$ is group-circulant]
        First, we establish that the matrix $\calM$ is indeed circulant with respect to the action of $\bbS_N$. This is a consequence of the fact that the uniform distribution over subsets of size $m$ is invariant under basis permutations. Based on \eqref{eq:exactMentries}, the matrix entries $\calM_{\alpha,\beta}$ only depend on the Hamming distance between the two subsets $\alpha,\,\beta$, and thus only on $\abs{\alpha\cup\beta}$. In fact, any matrix on this space whose entries depend only on the subset union (or intersection) size will be group-circulant. 
        
        The reason for this is that the $\bbS_N$ action is distance-transitive, i.e. invariant under diagonal action\footnote{Distance-transitivity has to do with the fact that $(\bbS_N,\;\bbS_t\times\bbS_{N-t})$ is not just a Gelfand pair, but a {\em symmetric} Gelfand pair \cite{ceccherini2008}. In the language of association schemes, it is related to the fact that the Johnson graph is not just distance-regular, but also distance-transitive.}. For subsets $\alpha, \beta \in \binom{[N]}{t}$, we have that $\abs{\alpha \cup \beta} = \abs{\alpha\beta^{-1} \cup [t]}$, where the multiplication is understood as permutation composition (which applies for any $\alpha,\beta$ coset representatives of $G/K$). Another way of seeing this is that the circulant function must be constant over the double-$K$-cosets, which in this setting are indexed by the subset distance, or equivalently by the number of items exchanged between $[t]$ and $[N]\setminus [t]$ by a corresponding permutation (item 4 of fact \ref{fact:johnsonfacts}). Consequently, as a function of the distance $p\in\{0,\dots, t\}$, the circulant function of the matrix $\calM$ is:
        \begin{equation}\label{eq:circulantfunctionM}
            \nu(p) = \left(\frac{m}{N}\right)^p\,\left(1 + O\left(\frac{t^2}{m}\right)\right)\,.
        \end{equation}
    \end{step}
    
    \begin{step}[Diagonalizing the matrix $\calM$ in the $\bbS_N$-Fourier basis]
         Since the matrix $\calM\in\C^{X\times X}$ is circulant, we can now apply fact \ref{lemma:diagonalization} with the specific irrep details from fact \ref{fact:johnsonfacts} to diagonalize it. This gives us the block-diagonal decomposition over the relevant two-row Young diagrams describing the allowed $\bbS_N$-irreps:
         \begin{align}
             W\calM W^\dagger &= \bigoplus_{q=0}^t \tilde\mu_{[N-q,q]}\,I_{d_{[N-q,q]}}\,.
         \end{align}
    Diagonalizing $\calM$ on the subspace $\Sunique=L(X)$ allows us to easily express the trace distance between our density matrix $\rho$ and the maximally mixed state $\rho_0$ under restriction to this specific subspace, namely:
    \begin{align}\label{eq:blockcontributions}
        \TD{\Pi_\text{unique}\,\rho\,\Pi_\text{unique},\;\Pi_\text{unique}\,\rho_0\,\Pi_\text{unique}} &= \frac{1}{2\binom{N+t-1}{t}}\norm{\calM - I_{\abs{X}}}_1 \\
        &= \frac{1}{2\binom{N+t-1}{t}}\sum_{q=0}^t d_{[N-q,q]}\;\abs{\tilde\mu_{[N-q,q]} - 1}\,.
    \end{align}
    The multiplicity $d_{[N-q,q]}$ of each eigenvalue $\tilde\mu_{[N-q,q]}$ is the dimension of the corresponding symmetric group irrep, which can be evaluated by the hook-length formula. In the case of two-row diagrams, the result is particularly simple:
        \begin{align}
            d_{[N-q,q]} &=\left\{ \begin{array}{lr}\binom{N}{q} - \binom{N}{q-1} & 1\leq q\leq t\\
            1 & q=0
            \end{array}
            \right.\,.
        \end{align}
    The eigenvalues are given by fact \ref{lemma:diagonalization} in terms of the circulant function \eqref{eq:circulantfunctionM} of the matrix $\calM$ and the spherical functions \eqref{eq:sphericalfunction} for this $(G, K)$ Gelfand pair. Since both depend only on the subset distance $p \in \{0,\dots, t\}$, we can reformulate this as a sum over distances; note that we have a counting factor of $\binom{t}{p}\binom{N-t}{p}$, which is the number of subsets of size $t$ at a distance $p$ from a reference subset. This gives us the eigenvalues as:
    \begin{align}
       \tilde\mu_{[N-q,q]} &= \sum_{p=0}^t \nu(p)\;\binom{t}{p}\binom{N-t}{p}\;\Phi_{[N-q,q]}(p)\,.
    \end{align}
    \end{step}
    
    \begin{step}[Diagonal contribution of the largest irrep block]
        We will show that it is enough to consider the most sizable eigenblock, corresponding to the Young diagram $[N-t,t]$. We will prove good agreement with the identity on this sub-subspace, which nonetheless fills most of the symmetric subspace.
        
        For this eigenvalue, the calculation simplifies significantly. Using the explicit form of the spherical function \eqref{eq:sphericalfunction}, we have that:
    \begin{align}
        \tilde\mu_{[N-t,t]} &= \sum_{p=0}^t\nu(p)\;\binom{t}{p}\binom{N-t}{p}\sum_{k=0}^p (-1)^k\binom{p}{k}\frac{1}{1-\frac{k}{N-t+1}}\,.
    \end{align}
    The $k$-summation can be expressed in closed form with the help of the combinatorial identity\footnote{Notice that the $k$-sums, which originate from the spherical functions, provide a distance-weighting on what otherwise would be the sum of the entries in a single row of the matrix $\calM$; without these factors, the sum would be superpolynomially large. Intuitively, the spherical functions contain the information about the local geometry of the graph of subsets.}:
    \begin{equation}
        \sum_{k=0}^p\frac{(-1)^k}{1 - \frac{k}{r}}\binom{p}{k} = \frac{(-1)^p}{\binom{r-1}{p}}\quad\quad \text{ for }r\geq p+1\,.
    \end{equation}
    With the substitution $r=N-t+1$, we obtain the particularly simple result:
    \begin{align}
        \tilde\mu_{[N-t,t]} 
        &= \sum_{p=0}^t (-1)^p \nu(p)\binom{t}{p} & \\
        &= \sum_{p=0}^t (-1)^p \left(\frac{m}{N}\right)^p\;\binom{t}{p}\left(1 + O\left(\frac{t^2}{m}\right)\right) & \text{(using the explicit circulant \eqref{eq:circulantfunctionM})}\,.
    \end{align}
    Finally, we notice the binomial sum in $p$, which gives us the desired asymptotic:
    \begin{align}
        \tilde\mu_{[N-t,t]} 
        &= \left(1 - \frac{m}{N}\right)^t\;\left(1 + O\left(\frac{t^2}{m}\right)\right) \\
        &= 1 + O\left(\frac{tm}{N}\right) +  O\left(\frac{t^2}{m}\right)\,.
    \end{align}
    \end{step}

    \begin{step}[Proximity to identity on the largest block is enough]
        Letting the projector onto this eigenblock be $\Pi_{[N-t,t]}$, our calculation of the eigenvalue gives us another restricted trace distance which further refines \eqref{eq:blockcontributions}:
    \begin{equation}
        \TD{\Pi_{[N-t,t]}\,\rho\,\Pi_{[N-t,t]},\;\Pi_{[N-t,t]}\,\rho_0\,\Pi_{[N-t,t]} }\leq O\left(\frac{tm}{N}\right) +  O\left(\frac{t^2}{m}\right)\,.
    \end{equation}
    However, this eigenblock covers a vast majority of all the dimensions in the symmetric subspace, since by birthday asymptotics:
    \begin{align}
        \frac{d_{[N-t,t]}}{\text{dim}\,\Sym{t}{[N]}} &= \frac{\binom{N}{t} - \binom{N}{t-1}}{\binom{N+t-1}{t}}\\
        &= 1 + O\left(\frac{t^2}{N}\right)\,.
    \end{align}
    By lemma \ref{lemma:nearbymatrices} (stated in the next section), we show that these two facts are enough to bound the total trace distance to the maximally mixed state on the full symmetric subspace. The intuition has to do with the fact that a dominant subspace will contain most of the probability mass of the maximally mixed state; if we constrain our unknown density matrix $\rho$ to be close to maximally mixed on this subspace, there is not enough probability mass left in $\rho$ to redistribute outside of this subspace in order to have a large impact on the overall trace distance.
    
    Formally, we use lemma \ref{lemma:nearbymatrices} with the choice of $H$ being the full symmetric subspace and $H_1$ being the subspace of the $[N-t,t]$ irrep. We have that the lemma parameters are $\epsilon=O(t^2/N)$ and $\delta=O(tm/N)+O(t^2/m)$ based on the above calculation. The lemma then gives us the overall proximity:
    \begin{align}
        \TD{\rho, \rho_0} \leq O\left(\frac{tm}{N}\right) + O\left(\frac{t^2}{m}\right)\,,
    \end{align}
    which is precisely what we aimed to bound. This concludes the proof of theorem \ref{thrm:randomsubsetstates}.\hfill\qed
    \end{step}

\subsection{A helpful lemma about nearby density matrices}\label{sec:helpfullemma}

In the proof of the above result, we made use of the following lemma about the case when a density matrix is close to maximally mixed on a dominant subspace:

\begin{lemma}\label{lemma:nearbymatrices}
    Consider the Hilbert space $H = H_1 \oplus H_2$ of dimension $D = d_1 + d_2$, where the two dimensions are uneven such that: 
    \begin{equation}
        \frac{d_2}{D} \leq \epsilon\,.
    \end{equation}
    Let $\rho$ be a density matrix on $H_1\oplus H_2$ such that the $H_1$-block of $\rho$ is close to the $H_1$-block of the maximally mixed state:
    \begin{equation}
        \TD{\Pi_1\rho\Pi_1, \frac{\Pi_1}{D}} \leq  \delta.
    \end{equation}
    In the above, $\Pi_1$ denotes the projector onto $H_1$. Then we have that the full density matrix $\rho$ is close to the full maximally mixed state:
    \begin{equation}
        \TD{\rho, \frac{I}{D}} \leq 2\delta + 2\epsilon\,.
    \end{equation}
\end{lemma}

\begin{proof}
    The proof follows from first principles and is delegated to appendix \ref{sec:proofoflemmanearbymatrices}.
\end{proof}

\section{Generalization to subset-phase states with biased phases}\label{sec:biasedphases}

Here, we briefly mention a natural extension of our result which interpolates between the random subset-phase states of Aaronson et al. \cite{aaronson2022quantum} and the subset states which we analyze in theorem \ref{thrm:randomsubsetstates}.
The basic idea is to consider random subset-phase state where the binary phases are \emph{biased} towards $+1$. 
Specifically, let us study the family of subset-phase states for which the phase function $f:[N]\to \{0,1\}$ is sampled independently at random for every argument $x\in[N]$ from a biased coin flip, such that $\mathbb{P}_f[f(x)=1] = \frac{1+b}{2}$. By modifying the bias\footnote{The regime $b < 0$ is equivalent by an overall sign flip symmetry.} parameter $b\in[0,1]$, we interpolate between the random subset-phase states (for which $b=0$, i.e. equal odds of $\pm 1$ signs) and the subset states ($b= 1$, i.e. phase is always $+1$). Because the values $f(x)$ are independent across $x$, the probability distribution is invariant under permuting the basis labels by $\bbS_N$ action, and the technique behind theorem \ref{thrm:randomsubsetstates} will still apply. Formally, we confirm this in the form of the following corollary:

\begin{corollary}[Biased subset-phase states]
\label{cor:biasedphasesubset}
        Let $t=O(\poly(n))$, let $m$ be in the range $\omega(\poly(n)) < m < o\left(2^n\right)$, and let $N=2^n$. Consider random subset-phase states generated by choosing the subset $S$ uniformly at random from $\binom{[N]}{m}$, and the phases at random from i.i.d. biased Bernoulli distributions such that $\mathbb P_f[f(x)=1]=\frac{1+b}{2}$, for some $b\in[-1, 1]$. Then we have that:
        \begin{equation}\label{eq:biasedtracedistance}
            \TD{\E_{\substack{
            S\sim \binom{[N]}{m}\\f\sim \Ber\left[\frac{1+b}{2}\right]^{[N]}
            }}\dyad{S, f}^{\otimes t},\;\E_{\ket\phi\sim\Haar{[N]}}\dyad{\phi}^{\otimes t}}\leq O\left(\frac{tmb^2}{N}\right) + O\left(\frac{t^2}{m}\right)\,.
        \end{equation}
\end{corollary} 

This corollary simultaneously generalizes our main theorem and the trace distance calculation of Aaronson et al. for unbiased phases \cite{aaronson2022quantum}.

\begin{proof}
    The argument from the proof of theorem \ref{thrm:randomsubsetstates} carries through almost identically. The only difference comes in evaluating the restricted density matrix entries in the $\Sunique$ basis. Instead of \eqref{eq:exactMentries}, the decay rate of the matrix entries will pick up a factor of the Bernoulli average $\E_f (-1)^{f(x)}=-b$, squared:
    \begin{align}
        \mel{\alpha}{\rho_b}{\beta} &= \frac{t!}{m^t}\mathbb{P}_{S\sim\binom{[N]}{m}}\left[\alpha\cup\beta\subset S\right]\; \E_{f\sim \Ber\left[\frac{1+b}{2}\right]^{[N]}}(-1)^{\sum_{a\in \alpha}f(a) + \sum_{b\in \beta}f(b)}\\
        &= \frac{t!}{m^t}\frac{\binom{N - \abs{\alpha\cup\beta}}{m - \abs{\alpha\cup\beta}}}{\binom{N}{m}} \;\left(\E_{z \sim\Ber\left[\frac{1+b}{2}\right]}(-1)^z\right)^{2(\abs{\alpha\cup\beta}-t)}\\
        &=\frac{1}{\binom{N+t-1}{t}}\left(\frac{mb^2}{N}\right)^{\abs{\alpha\cup\beta} - t}\;\left(1 + \left(\frac{t^2}{m}\right)\right)\,.
    \end{align}
    This has the effect of replacing the matrix entry decay rate $m/N$ with $mb^2/N$ from the point of view of the spectral analysis in the proof of theorem \ref{thrm:randomsubsetstates}, which otherwise continues unaltered.
\end{proof}

Interestingly, this means that if one wishes to use subsets of very large size (say polynomially close to $2^n$) in a pseudorandomness construction, one must compensate with many phases (bias $b=o(1)$). As before, this can be seen as a consequence of the fact the inner product with the $\ket{+^n}$ state must remain negligible to avoid detection by the swap test (i.e. projective measurement) against $\ket{+^n}$.

\section{Discussion and open problems}

This work centered on proving theorem \ref{thrm:randomsubsetstates}, which is the information-theoretic statement that random subset states are information theoretically close to the Haar measure.
Replacing the true randomness with quantum-secure pseudorandomness along the standard hybrid argument in \cite{aaronson2022quantum}, we obtain a pseudorandom, pseudoentangled state ensemble, which is computationally indistinguishable from Haar randomness. This removes the need for random relative phases which played a key role in previous pseudorandom state constructions.

The main technical tools used lies in the application of representation-theoretic methods in the calculation of the trace distance between the average over subset states and the Haar average, in the limit in which the number of copies $t=\poly(n)$ is significantly smaller than the subset size $m$, which is itself significantly smaller than the local dimension $N=2^n$. When averaging over a uniform distribution over $[N]$-subsets of size $m$, it is not surprising that the resulting object has a high degree of (permutation) symmetry. However, it is a distinct feature of the specific symmetries of this problem that we encounter typicality --- such that the largest irrep block is also the largest eigenblock, and also describes most of the problem.

While we approach this problem from the perspective of the representation theory of the symmetric group, we note that an equivalent analysis can be performed with tools from algebraic combinatorics. In particular, the specific spherical functions for this homogeneous space were initially obtained from the study of the association scheme known as the {\em Johnson scheme} in the context of coding theory \cite{delsarte1975}. For example, the specific irrep structure of $\mathrm{Ind}_{\bbS_{t} \times \bbS_{N-t}}^{\bbS_N}1$ describing this problem is identical to the spectral structure of the adjacency matrix of the so-called {\em Johnson graph} $J(N,t)$ --- the natural distance-regular graph defined with subsets of size $t$ as nodes, and whose edges connect subsets which differ by a single element. This is a canonical instance of a more general correspondence: in the group-free setting, such distance-regular graphs and their related association schemes exhibit similar algebraic properties as homogeneous spaces over Gelfand pairs in the group-theoretic setting. For a pedagogical summary of this correspondence, see for example the monograph \cite{ceccherini2008}.

Our work leaves open a number of problems: 
\begin{enumerate}
    \item Are pseudorandom subset states a simpler cryptographic object than pseudorandom subset phase states? Note that the preparation of pseudorandom subset phase states uses both a pseudorandom function (PRF) and a pseudorandom permutation (PRP), while the construction of pseudorandom subset states only uses the PRP. Both of these are of course equivalent to one way functions, so the question is whether there is a simpler primitive than PRPs which can efficiently generate pseudorandom subset states. In other words, one can ask whether the input-output security of pseudorandom permutations is strictly needed for PRS preparation, or whether a meaningful, weaker notion of security suffices.\footnote{We thank Mark Zhandry for suggesting this question.}
    
    \item Can we extend our main result (Theorem \ref{thrm:randomsubsetstates}) to other distributions over subsets, i.e. not just uniformly random subsets of a particular size?  Uniformly random subsets generate states with very simple entanglement structures --- the entanglement across all cuts is simultaneously low, so the states are in some sense ``geometry-free''. In contrast, many physical systems have spatial entanglement structures like area-law or volume-law entanglement. Generalizing our result to other distributions over subsets might allow one to create pseudorandomness with differing entanglement structures which are more physically relevant. 
\end{enumerate}

\section*{Acknowledgments}
We thank Roozbeh Bassirian, Soumik Ghosh, Fermi Ma, Alex May, Tony Metger, Henry Yuen, Mark Zhandry, and Chenyi Zhang for helpful discussions.
A.B. and T.G.T. were supported
in part by the U.S. DOE Office of Science under Award Number DE-SC0020266.
A.B. was supported in part by the DOE QuantISED grant DE-SC0020360 and by the AFOSR under grant FA9550-21-1-0392. 

\vspace{20pt}
\bibliographystyle{alpha}
\bibliography{main}

\newcommand{\etalchar}[1]{$^{#1}$}
\begin{thebibliography}{NZO{\etalchar{+}}21}

\bibitem[ABF{\etalchar{+}}22]{aaronson2022quantum}
Scott Aaronson, Adam Bouland, Bill Fefferman, Soumik Ghosh, Umesh Vazirani,
  Chenyi Zhang, and Zixin Zhou.
\newblock Quantum pseudoentanglement.
\newblock {\em arXiv preprint
  \href{https://arxiv.org/abs/2211.00747}{arXiv:2211.00747}}, 2022.
\newblock To appear in Proc. ITCS'24.

\bibitem[AC12]{aaronson2012quantum}
Scott Aaronson and Paul Christiano.
\newblock Quantum money from hidden subspaces.
\newblock In {\em Proceedings of the forty-fourth annual ACM symposium on
  Theory of computing}, pages 41--60, 2012.
\newblock \href{https://arxiv.org/abs/1203.4740}{arXiv:1203.4740}.

\bibitem[BNZZ19]{bannai2019explicit}
Eiichi Bannai, Mikio Nakahara, Da~Zhao, and Yan Zhu.
\newblock On the explicit constructions of certain unitary $t$-designs.
\newblock {\em Journal of Physics A: Mathematical and Theoretical},
  52(49):495301, 2019.
\newblock \href{https://arxiv.org/abs/1906.04583}{arXiv:1906.04583}.

\bibitem[BS19]{brakerski2019pseudo}
Zvika Brakerski and Omri Shmueli.
\newblock (pseudo) random quantum states with binary phase.
\newblock In {\em Theory of Cryptography Conference}, pages 229--250. Springer,
  2019.
\newblock \href{https://arxiv.org/abs/1906.10611}{arXiv:1906.10611}.

\bibitem[CSST08]{ceccherini2008}
Tullio Ceccherini-Silberstein, Fabio Scarabotti, and Filippo Tolli.
\newblock {\em {Harmonic Analysis on Finite Groups: Representation Theory,
  Gelfand Pairs and Markov Chains}}.
\newblock {Cambridge Studies in Advanced Mathematics}. Cambridge University
  Press, 2008.

\bibitem[Del75]{delsarte1975}
Phillippe Delsarte.
\newblock {The Association Schemes of Coding Theory}.
\newblock In {\em Combinatorics}, pages 143--161. Springer Netherlands, 1975.

\bibitem[Dia88]{diaconis1988group}
Persi Diaconis.
\newblock {\em Group representations in probability and statistics}.
\newblock Institute of Mathematical Statistics Lecture Notes --- Monograph
  Series, 11. Institute of Mathematical Statistics, 1988.

\bibitem[FK18]{fefferman2015quantum}
Bill Fefferman and Shelby Kimmel.
\newblock {Quantum vs. Classical Proofs and Subset Verification}.
\newblock In {\em 43rd International Symposium on Mathematical Foundations of
  Computer Science (MFCS 2018)}, volume 117 of {\em Leibniz International
  Proceedings in Informatics (LIPIcs)}, pages 22:1--22:23, 2018.
\newblock \href{https://arxiv.org/abs/1510.06750}{arXiv:1510.06750}.

\bibitem[GKS15]{grilo2015qma}
Alex~Bredariol Grilo, Iordanis Kerenidis, and Jamie Sikora.
\newblock Qma with subset state witnesses.
\newblock In {\em International Symposium on Mathematical Foundations of
  Computer Science}, pages 163--174. Springer, 2015.
\newblock \href{https://arxiv.org/abs/1410.2882}{arXiv:1410.2882}.

\bibitem[Har13]{harrow2013church}
Aram~W Harrow.
\newblock The church of the symmetric subspace.
\newblock {\em arXiv preprint
  \href{https://arxiv.org/abs/1308.6595}{arXiv:1308.6595}}, 2013.

\bibitem[HBK23]{haug2023pseudorandom}
Tobias Haug, Kishor Bharti, and Dax~Enshan Koh.
\newblock Pseudorandom unitaries are neither real nor sparse nor noise-robust.
\newblock {\em arXiv preprint
  \href{https://arxiv.org/abs/2306.11677}{arXiv:2306.11677}}, 2023.

\bibitem[JLS18]{ji2018pseudorandom}
Zhengfeng Ji, Yi-Kai Liu, and Fang Song.
\newblock Pseudorandom quantum states.
\newblock In {\em Advances in Cryptology--CRYPTO 2018: 38th Annual
  International Cryptology Conference, Santa Barbara, CA, USA, August 19--23,
  2018, Proceedings, Part III 38}, pages 126--152. Springer, 2018.
\newblock \href{https://eprint.iacr.org/2018/544}{eprint.iacr.org/2018/544}.

\bibitem[Ma23]{macommunication}
Fermi Ma.
\newblock December 2023.
\newblock Personal communication.

\bibitem[NPS23]{narayanan2023multiset}
Sridhar Narayanan, Digjoy Paul, and Shraddha Srivastava.
\newblock {The Multiset Partition Algebra}.
\newblock {\em Israel Journal of Mathematics}, 255(1):453--500, 2023.
\newblock \href{https://arxiv.org/abs/1903.10809}{arXiv:1903.10809}.

\bibitem[NZO{\etalchar{+}}21]{nakata2021quantum}
Yoshifumi Nakata, Da~Zhao, Takayuki Okuda, Eiichi Bannai, Yasunari Suzuki,
  Shiro Tamiya, Kentaro Heya, Zhiguang Yan, Kun Zuo, Shuhei Tamate, et~al.
\newblock Quantum circuits for exact unitary $t$-designs and applications to
  higher-order randomized benchmarking.
\newblock {\em PRX Quantum}, 2(3):030339, 2021.
\newblock \href{https://arxiv.org/abs/2102.12617}{arXiv:2102.12617}.

\bibitem[OZ23]{orellana2023howe}
Rosa Orellana and Mike Zabrocki.
\newblock Howe duality of the symmetric group and a multiset partition algebra.
\newblock {\em Communications in Algebra}, 51(1):393--413, 2023.
\newblock \href{https://arxiv.org/abs/2007.07370}{arXiv:2007.07370}.

\bibitem[Wat00]{watrous2000succinct}
John Watrous.
\newblock Succinct quantum proofs for properties of finite groups.
\newblock In {\em Proceedings 41st Annual Symposium on Foundations of Computer
  Science}, pages 537--546. IEEE, 2000.
\newblock \href{https://arxiv.org/abs/cs/0009002}{arXiv:cs/0009002}.

\bibitem[Zha16]{zhandry2016note}
Mark Zhandry.
\newblock A note on quantum-secure prps.
\newblock {\em arXiv preprint
  \href{https://arxiv.org/abs/1611.05564}{arXiv:1611.05564}}, 2016.

\end{thebibliography}

\vspace{20pt}
\appendix
\section{Proof of lemma \ref{lemma:nearbymatrices}}\label{sec:proofoflemmanearbymatrices}
Let the block decomposition of $\rho_1$ along $H_1\oplus H_2$ be:
    \begin{equation}
        \rho = \begin{pmatrix}
            J & C \\
            C^\dagger & E
        \end{pmatrix}\,.
    \end{equation}
    Define the matrix:
    \begin{equation}
        \Delta = \begin{pmatrix}
            0 & C \\
            C^\dagger & E-I_2/D
        \end{pmatrix}\,.
    \end{equation}
    By rank subadditivity, we have that the rank of $\Delta$ is at most:
    \begin{align}
        \rank\;\Delta &\leq \rank\begin{pmatrix}
            0 & C \\
            0 & 0
        \end{pmatrix} + \rank \begin{pmatrix}
            0 & 0 \\
            C^\dagger & E - I_2/D
        \end{pmatrix} \\
        & \leq d_2 + d_2 = 2 d_2\,.
    \end{align}
    Let the eigenvalues of $\Delta$ be:
    \begin{equation}
        \text{spec}(\Delta) = (\lambda_1,\dots,\lambda_k,\underbrace{0,0,\dots,0}_{D-k\text{ times }})\,,
    \end{equation}with the number $k$ of nonzero eigenvalues equal to the rank $0 \leq k \leq 2d_2$.

    Then the eigenvalues of $\Delta + I/D$ are $\lambda_1+1/D,\dots,\lambda_k+1/D,1/D,1/D,\dots, 1/D$. Denote the trace norm (i.e. the Schatten 1-norm) by $\norm{\cdot}_1$, and notice that since $\rho$ is a density matrix we have that $\norm{\rho}_1=1$.
    
    Using the premise and the triangle inequality, we get an upper bound on the trace norm of $\Delta+I/D$ as:
    \begin{align}
        \abs{\lambda_1+\frac{1}{D}} + \dots + \abs{\lambda_k + \frac{1}{D}} + \frac{D-k}{D} &= \norm{\Delta + I/D}_1 \\
        &\leq \norm{\rho}_1 + \norm{\begin{pmatrix}
            I_1/D & 0 \\
            0 & 0
        \end{pmatrix} - \begin{pmatrix}
            J & 0 \\
            0 & 0
        \end{pmatrix}}_1\\
        &\leq 1 + 2\delta\,.
    \end{align}
    The triangle inequality also gives us a lower bound:
    \begin{align}
        \abs{\lambda_1+\frac{1}{D}} + \dots + \abs{\lambda_k + \frac{1}{D}} &\geq \abs{\lambda_1} +\dots+ \abs{\lambda_k} - \frac{k}{D}\\
        &= \norm{\Delta}_1 - \frac{k}{D}\,.
    \end{align}
    Combining the above two, we get that:
    \begin{align}
        \norm{\Delta}_1 &\leq 2\delta + \frac{2k}{D}\\
        &\leq 2\delta + 4\epsilon\,.
    \end{align}
    This is enough for the trace distance we want, since:
    \begin{align}
        \TD{\rho, I/D} &= \frac{1}{2}\norm{\rho - I/D}_1 \\
        &\leq \frac12\norm{\Delta}_1 + \frac12\norm{\begin{pmatrix}
            I_1/D - J & 0 \\
            0 & 0
        \end{pmatrix}}_1\\
        &\leq 2\delta + 2\epsilon\,.
    \end{align}
    \hfill\qed

\end{document}